\documentclass[runningheads]{llncs}
\usepackage[utf8]{inputenc}
\usepackage[T1]{fontenc}
\usepackage{microtype}
\usepackage[dvips]{graphicx}
\usepackage{xcolor}
\usepackage{times}
\usepackage{enumitem}
\usepackage{graphicx}
\usepackage{amssymb}
\usepackage{mathrsfs}
\usepackage{maple}
\usepackage{mathtools}
\usepackage{amsmath}
\usepackage{hyperref}
\usepackage{microtype}
\usepackage{algorithm,algpseudocode}
\usepackage{algorithmicx}
\usepackage{caption}
\usepackage{listings}
\usepackage{orcidlink}
\usepackage{breqn}
\usepackage{textcomp}
\usepackage{url}
\usepackage{moreverb}
\usepackage{verbatim}
\usepackage{sverb}
\usepackage{geometry}
\DeclareMathOperator{\arctanh}{arctanh}

\DeclareMathOperator{\arcsinh}{arcsinh}

\newcommand*{\QEDA}{\hfill\ensuremath{\blacksquare}}

\begin{document}
	\title{On the representation of non-holonomic univariate power series}

	\author{Bertrand Teguia Tabuguia \orcidlink{0000-0001-9199-7077}~  \and Wolfram Koepf}
	
	\authorrunning{B. Teguia Tabuguia and W. Koepf}
	
	\institute{University of Kassel, Heinrich-Plett-Str. 40. 34132 Kassel\\
		\email{\{bteguia,koepf\}@mathematik.uni-kassel.de}
	}
	
	\maketitle        

	\begin{abstract}
		Holonomic functions play an essential role in Computer Algebra since they allow the application of many symbolic algorithms. Among all algorithmic attempts to find formulas for power series, the holonomic property remains the most important requirement to be satisfied by the function under consideration. The targeted functions mainly summarize that of meromorphic functions. However, expressions like $\tan(z)$, $z/(\exp(z)-1)$, $\sec(z)$, etc., particularly, reciprocals, quotients and compositions of holonomic functions, are generally not holonomic. Therefore their power series are inaccessible by the holonomic framework, including Maple's \texttt{convert/FormalPowerSeries} command up to Maple 2021. From the mathematical dictionaries, one can observe that most of the known closed-form formulas of non-holonomic power series involve another sequence whose evaluation depends on some finite summations. In the case of $\tan(z)$ and $\sec(z)$ the corresponding sequences are the {B}ernoulli and {E}uler numbers, respectively. Thus providing a symbolic approach that yields complete representations when linear summations for power series coefficients of non-holonomic functions appear, might be seen as a step forward towards the representation of non-holonomic power series.
		
		 By adapting the method of ansatz with undetermined coefficients, we build an algorithm that computes least-order quadratic differential equations with polynomial coefficients for a large class of non-holonomic functions. A differential equation resulting from this procedure is converted into a recurrence equation by applying the Cauchy product formula and rewriting powers into polynomials and derivatives into shifts. Finally, using enough initial values we are able to give normal form representations (Geddes et al. 1992) to characterize several non-holonomic power series. As a consequence of the defined normal transformation, it turns out that our algorithm is able to detect identities between non-holonomic functions that were not accessible in the past. We discuss this algorithm and its implementation for Maple 2022.
		
		Our Maple and Maxima implementations are available under the \texttt{FPS} software which can be downloaded at\\ \url{http://www.mathematik.uni-kassel.de/~bteguia/FPS_webpage/FPS.htm}.

	\end{abstract}
	
	\section{Introduction}
	
	Let $\mathbb{K}$ be a field of characteristic zero; mostly $\mathbb{K}$ denotes a finite extension field of the rationals. A function $f(z)$ is holonomic (or $D$-finite) over $\mathbb{K}$, if it satisfies a homogeneous linear differential equation with polynomial coefficients in $\mathbb{K}[z]$. Similarly, a sequence $(a_n)_{n\geqslant0}$ of numbers in $\mathbb{K}$ is holonomic (or $P$-recursive) over $\mathbb{K}$, if it satisfies a homogeneous linear recurrence equation with polynomial coefficients in $\mathbb{K}[n]$. Since analytic functions can be represented by power series, holonomic power series connect analytic holonomic functions to holonomic sequences (see \cite{stanley1980differentiably}). However, in this paper, we deal with ``formal'' power series, i.e., our computations are independent of any notion of convergence, and are valid for a given function at a point of expansion $z_0$ $\in\mathbb{K} \cup \{-\infty,\infty\}$, whenever the analytic requirement at $z_0$ is guaranteed. For implementation purposes, the Maple \texttt{series} command enables us to authenticate the existence of a series expansion before we proceed to search for a representation for it.
	
	The exponential generating function of the {B}ernoulli numbers $B_n,n=0,1,\ldots,$ has the series formula
	\begin{equation}
		\frac{z}{\exp(z)-1} = \sum_{n=0}^{\infty} \frac{B_n}{n!} z^n. \label{eq1}
	\end{equation}
	
	Numerous recurrence equations are known for the {B}ernoulli numbers (see \cite{agoh2009shortened}). The most basic among them is
	\begin{equation}
		\sum_{k=0}^{n-1} \binom{n}{k} B_k = 0,~ n\geqslant 2 \label{eq2}
	\end{equation}
	from which one can compute the series coefficients in $(\ref{eq1})$ from the initial value $B_0=1$. All explicit recurrence formulas of the {B}ernoulli numbers present the appearance of finite summation(s) that we can see as an indicator of its non-holonomic property. A proof that the sequence of {B}ernoulli numbers is not holonomic can be found in \cite{carlitz1964recurrences}. This fact implies that their generating functions as well as functions whose power series are defined by means of the sequence of {B}ernoulli numbers are non-holonomic (see \cite{jungen1931series}, \cite{mallinger1996}). An important reference for the non-holonomicity of $\tan(z),$ $\csc(z)$, and $\sec(z),$ is \cite{stanley1980differentiably}. We call the series of such functions non-holonomic power series. Note that the same conclusions also arise for power series whose coefficients depend on the {E}uler numbers since they are connected to the {B}ernoulli numbers. Furthermore, note that the non-holonomic character is not limited to the presence of the Bernoulli and Euler numbers. It is well known that quotients and compositions of holonomic functions are generally not holonomic. Other non-holonomic sequences of numbers can be at the core definition of several other non-holonomic power series. 
	
	We recall that the ``\textit{form}'' level abstraction considers functions as they are represented in terms of ``chosen'' elementary functions. It recognizes that a particular function can have many different valid representations in terms of these elementary functions. This is why we often differentiate between an expression which refers to the form, and a function which refers to the the mathematical object with its range and its domain of definition. Many expression classes allow simplification from one expression to another using a ``\textit{normal transformation}\footnote{In the cited reference, the authors use ``normal function''. We do not use this designation to avoid confusion with the objects we manipulate.}'' which can prove the zero-equivalence of their difference. What we call a normal form in a given class of expressions is a representation that is invariant under any application of the normal transformation used in that class. Unlike a normal transformation, a canonical transformation always gives the same representation for equivalent expressions. For more details on normal forms and algebraic representations, see \cite[Chapter 3]{geddes1992algorithms}.
	
	This paper is concerned with a general-purpose symbolic algorithm that defines a normal transformation for a class above holonomic functions by computing normal forms of their power series. Moreover, in several cases our algorithm is able to find the same representation for equivalent non-holonomic expressions, making it behaves like a canonical transformation. The non-holonomic power series represented by our algorithm have finite summations in \textit{quadratic recurrence equations} satisfied by their coefficients. These \textit{sum-recursive recurrence equations} (that we will often call quadratic recurrence equations) are similar to those considered in \cite{agoh2007convolution} for finding convolution identities for {B}ernoulli numbers. With this approach we are able to recover some Ramanujan identities mentioned in \cite{chellali1988acceleration}. We deduce quadratic recurrence equations from homogeneous quadratic differential equations (QDE) with polynomial coefficients. These can be seen as higher-order differential equations than the ones considered in \cite{chalkley1960second} with some linear differential monomial appearing (see also \cite[Section 6]{eremenko1982meromorphic}). Nevertheless, the differential equation sought is of least order possible. Our computational method for finding QDEs proceeds like that used in \cite{schatz2019automatic}, but in a less direct and more efficient way (avoidance of nested loops) which also allows us to define more precisely the type of functions we expect as inputs. Remark that this differs our development from recent computations of non-linear differential equations for {B}ernoulli numbers; as for instance \cite{choi2014note} and \cite{khan2019finding}.
	
	As we use quadratic recurrence equations (QRE) to give power series formulas, the targeted representations are given in a recursive form. This allows us to obtain the following representations using our Maple implementation.
	
	\begin{example}\item \vspace{-0.2cm}
		
		\begin{small}
			\begin{lstlisting}
				> FPS(z/(exp(z)-1),z,n)
			\end{lstlisting}
			
			\vspace{-0.3cm}
			
			\begin{dmath}\label{eq3}
				\mathit{Series} \! \left(\left[\moverset{\infty}{\munderset{n =0}{\textcolor{gray}{\sum}}}\! A \! \left(n \right) z^{n},
				A \! \left(n +3\right)\hiderel{=}-\frac{\left(\moverset{n +2}{\munderset{\textit{\_k} =1}{\textcolor{gray}{\sum}}}\! A \! \left(\textit{\_k} \right) A \! \left(n +3-\textit{\_k} \right)\right)+A \! \left(n +2\right)}{n +4}\right],
				\\[-0.33mm]
				\left\{A \! \left(n \right)\right\},\left\{A \! \left(0\right) \hiderel{=} 1,A \! \left(1\right) \hiderel{=} -\frac{1}{2},A \! \left(2\right) \hiderel{=} \frac{1}{12}\right\},\mathit{INFO} \right)
			\end{dmath}
			
			\begin{lstlisting}
				> FPS(1/log(1+z),z,n)
			\end{lstlisting}
			
			\vspace{-0.3cm}
			
			\begin{dmath}\label{eq4}
				\hspace{-0.8cm} \mathit{Series} \! \left(\left[\moverset{\infty}{\munderset{n =0}{\textcolor{gray}{\sum}}}\! A \! \left(n \right) z^{n -1}, 
				A \! \left(n +3\right)=-\frac{\left(n +1\right) A \! \left(n +2\right)+\left(\moverset{n +2}{\munderset{\textit{\_k} =1}{\textcolor{gray}{\sum}}}\! A \! \left(\textit{\_k} \right) A \! \left(n +3-\textit{\_k} \right)\right)}{n +4}\right],
				\\[-0.33mm]
				\left\{A \! \left(n \right)\right\},\left\{A \! \left(0\right) \hiderel{=} 1,A \! \left(1\right) \hiderel{=} \frac{1}{2},A \! \left(2\right) \hiderel{=} -\frac{1}{12}\right\},\mathit{INFO} \right).
			\end{dmath}
		\end{small}
	\end{example}
	
	As given above, all series expansions will be given in the neighborhood of zero since the general case can easily be deduced from this. The result of this paper is an improved version of one contribution from first author's Ph.D. thesis (see \cite[Section 8.7]{BTphd}) which is concerned with symbolic computation of formal power series, hence the acronym \texttt{FPS} used above. The thesis deals with the holonomic, hypergeometric type, and non-holonomic functions. For holonomic and hypergeometric type functions, see the references \cite{BTmc2020}, \cite{BThypergeometric}, and \cite{teguia2021symbolic}. Our software is available for download for Maple 2021 and Maxima 5.44 users at its dedicated web page \url{http://www.mathematik.uni-kassel.de/~bteguia/FPS_webpage/FPS.htm}.
	
	In Section \ref{sec2}, we specify the class of functions that are eligible to our method. Then in Section \ref{sec3}, we describe our algorithm for finding QDEs. This will be followed by explanations on how we find QREs and deduce series representations with appropriate numbers of initial values. We will also see how the whole algorithm is able to prove identities (see \cite[Section 3.3]{geddes1992algorithms}, \cite[Exercise 9.8]{koepf2021computer}) like
	
	\begin{equation}\label{eq5}
		\log\left(\tan\left(\dfrac{z}{2}\right) + \sec\left(\dfrac{z}{2}\right)\right) = \arcsinh\left(\dfrac{\sin(z)}{1+\cos(z)}\right),~~ -\pi < z < \pi.
	\end{equation}

	\section{The class of $\delta_2$-finite functions}\label{sec2}
	
	It is not precise to say that we consider non-holonomic functions given that we intend to describe an algorithm to compute their power series. Since algorithms only operate on finite data structures, we must define a suitable class of functions, which naturally contains holonomic functions, and extends to functions like $\sec(z)$, $\cot(z)$, $\csc(z)$, $\tan(z)$, etc. An analogous development is presented in \cite{kauers2011}. Starting from the observation that differentiating $y(z)=\tan(z)$ yields $y'(z)=1+\left(\tan(z)\right)^2=1+y(z)^2$, one can see the targeted class of functions as solutions of a certain type of algebraic ordinary differential equations.
	
	Throughout the paper, we denote by $\mathbb{K}$ a field of characteristic zero (generally $\mathbb{K}$ is a finite extension field of $\mathbb{Q}$), and we consider differential equations with rational coefficients over $\mathbb{K}$. We assume that 	
	
	\begin{equation}\label{eq6}
		\frac{d^{-1}}{dz}f=f^{(-1)}=1, ~\text{ and }~\frac{d^{0}}{dz}f=f^{(0)}=f,
	\end{equation}
	for a differentiable function $f$.
	
	\begin{definition}[Homogeneous quadratic differential equation]\label{def1} Let $d$ be a non-negative integer. A differential equation of order $d$ in the dependent variable $y$ is said to be homogeneous quadratic over $\mathbb{K}$, if there exist polynomials $P_0$, $P_1$,$\ldots$, $P_{r}$, $r = 1+d(d+5)/2$ (see Proposition \ref{prop}), such that
		\begin{equation}\label{eq7}
			P_r {y^{(d)}}^2 + P_{r-1} y^{(d)}y^{(d-1)} + \cdots + P_{r-d-1} y^{(d)} + \cdots + P_4 {y'}^2 + P_3 y'y+P_2 y'+ P_1 y^2 + P_0 y=0,
		\end{equation}
		and $P_r,\ldots,P_{r-d}$ are not all zero.
	\end{definition}
	
	In this definition, it is required that at least one of the polynomial coefficients of quadratic differential monomials\footnote{Quadratic differential monomial: terms with product of two derivatives (including the square of a derivative).} of order $d$ is non-zero. Thus the other coefficients might equal zero. This is equivalent to say that $(\ref{eq7})$ is a homogeneous quadratic differential equation of order $d$ if ``the polynomial coefficients of quadratic differential monomials of order at most $d$ are not all zero, and at least one of $P_r,\ldots,P_{r-d}$ is non-zero''. To admit holonomic functions, it suffices to change this condition to ``$P_r,\ldots,P_0$ are not all zero, and at least one of $P_r,\ldots,P_{r-d-1}$ is non-zero'', and this defines our class of functions. However, as holonomic functions are also called $D$-finite functions: where $D$ stands for ``differentiability'', and finite refers to the requirement that the order should be finite, we would like to similarly define our class of functions so that the algorithmic approach of Section \ref{sec3} differs from that of holonomic functions only by the differential operator used.
	
	Let $f(z)$ be a differentiable function. Consider the following scheme
	
	\begin{equation}
		\begin{matrix}
			(1)~ 1,&&&&\\
			(2)~ f,& (3)~ f^2,& & & \\
			(4)~ f',& (5)~ f'f,& (6)~ (f')^2,& & \\
			(7)~ f'',& (8)~ f''f,& (9)~ f''f',& (10)~ (f'')^2, & \\
			(11)~ f''',& (12)~ f'''f,& (13)~ f'''f',& (14)~ f'''f'',& (15)~ (f''')^2,\\
			\ldots&&&&
		\end{matrix}
		\label{eq8}
	\end{equation}
	and assume that the positive integers in parentheses represent the derivative orders of the derivative operator that we are looking for. This operator, say $\delta_{2,z}$, computes the product of two derivatives of $f$ according to the ordering given in $(\ref{eq8})$.
	
	Looking at $(\ref{eq8})$ as an infinite lower triangular matrix reduces the definition of $\delta_{2,z}$ to that of a one-to-one map $\nu$ between positive integers and the corresponding subset of $\mathbb{N}\times\mathbb{N}$: $(i,j)_{i,j\in\mathbb{N}}, i\leqslant j$. This can be done by counting the couple\footnote{Tuple of two elements.} $(i,j)$ in $(\ref{eq8})$ from up to down, from the left to the right. We obtain
	\begin{equation}\label{eq9}
		\nu(k) = (i,j) = \begin{cases} (l,l)~~ \text{ if }~~ N=k\\ (l+1,k-N)~\text{ otherwise}\end{cases},
		\text{ where } l=\left\lfloor \sqrt{2k+\frac{1}{4}}-\frac{1}{2}\right\rfloor,~\text{ and }~N=\dfrac{l(l+1)}{2}.
	\end{equation}
	It remains to define a correspondence between the couple $(i,j)=\nu(k),~k\in\mathbb{N}$ and the differential monomials in $(\ref{eq8})$. This is straightforward by considering the assumption of $(\ref{eq6})$. We get
	\begin{equation}\label{eq10}
		\delta_{2,z}^k(f) = \dfrac{d^{i-2}}{dz^{i-2}}f \cdot \dfrac{d^{j-2}}{dz^{j-2}}f,~\text{ where }~ (i,j) = \nu(k).
	\end{equation}
	
	We implemented this operator in our Maxima package as \texttt{delta2diff(f,z,k)}. One can use it to recover some products of derivatives in $(\ref{eq6})$.
	
	\begin{example}\item
		
		\noindent
		\begin{minipage}[t]{8ex}\color{red}\bf
			\begin{verbatim}
				(%i1) 
			\end{verbatim}
		\end{minipage}
		\begin{minipage}[t]{\textwidth}\color{blue}
			\begin{verbatim}
				delta2diff(f(z),z,3);
			\end{verbatim}
		\end{minipage}
		\vspace{-0.25cm}
		\definecolor{labelcolor}{RGB}{100,0,0}
		\[\displaystyle
		\hspace{1cm}\parbox{10ex}{$\color{labelcolor}\mathrm{\tt (\%o1) }\quad $}\hspace{1cm}
		{{\mathrm{f}\left( z\right) }^{2}}\mbox{}
		\]
		\vspace{-0.25cm}
		
		\noindent
		\begin{minipage}[t]{8ex}\color{red}\bf
			\begin{verbatim}
				(%i2) 
			\end{verbatim}
		\end{minipage}
		\begin{minipage}[t]{\textwidth}\color{blue}
			\begin{verbatim}
				delta2diff(f(z),z,4);
			\end{verbatim}
		\end{minipage}
		\vspace{-0.25cm}
		\definecolor{labelcolor}{RGB}{100,0,0}
		\[\displaystyle
		\hspace{1cm}\parbox{10ex}{$\color{labelcolor}\mathrm{\tt (\%o2) }\quad $}\hspace{1cm}
		\frac{d}{d\,z}\cdot \mathrm{f}\left( z\right) \mbox{}
		\]
		\vspace{-0.25cm}
		
		\noindent
		\begin{minipage}[t]{8ex}\color{red}\bf
			\begin{verbatim}
				(%i3) 
			\end{verbatim}
		\end{minipage}
		\begin{minipage}[t]{\textwidth}\color{blue}
			\begin{verbatim}
				delta2diff(f(z),z,5);
			\end{verbatim}
		\end{minipage}
		\vspace{-0.25cm}
		\definecolor{labelcolor}{RGB}{100,0,0}
		\[\displaystyle
		\hspace{1cm}\parbox{10ex}{$\color{labelcolor}\mathrm{\tt (\%o3) }\quad $}\hspace{1cm}
		\mathrm{f}\left( z\right) \cdot \left( \frac{d}{d\,z}\cdot \mathrm{f}\left( z\right) \right) \mbox{}
		\]
		\vspace{-0.25cm}
		
		\noindent
		\begin{minipage}[t]{8ex}\color{red}\bf
			\begin{verbatim}
				(%i4) 
			\end{verbatim}
		\end{minipage}
		\begin{minipage}[t]{\textwidth}\color{blue}
			\begin{verbatim}
				delta2diff(f(z),z,6);
			\end{verbatim}
		\end{minipage}
		\vspace{-0.25cm}
		\definecolor{labelcolor}{RGB}{100,0,0}
		\[\displaystyle
		\hspace{1cm}\parbox{10ex}{$\color{labelcolor}\mathrm{\tt (\%o4) }\quad $}\hspace{1cm}
		{{\left( \frac{d}{d\,z}\cdot \mathrm{f}\left( z\right) \right) }^{2}}\mbox{}
		\]
		\vspace{-0.25cm}
		
		\noindent
		\begin{minipage}[t]{8ex}\color{red}\bf
			\begin{verbatim}
				(%i5) 
			\end{verbatim}
		\end{minipage}
		\begin{minipage}[t]{\textwidth}\color{blue}
			\begin{verbatim}
				delta2diff(f(z),z,14);
			\end{verbatim}
		\end{minipage}
		\vspace{-0.25cm}
		\definecolor{labelcolor}{RGB}{100,0,0}
		\[\displaystyle
		\hspace{1cm}\parbox{10ex}{$\color{labelcolor}\mathrm{\tt (\%o5) }\quad $} \hspace{1cm}
		\left( \frac{{{d}^{2}}}{d\,{{z}^{2}}}\cdot \mathrm{f}\left( z\right) \right) \cdot \left( \frac{{{d}^{3}}}{d\,{{z}^{3}}}\cdot \mathrm{f}\left( z\right) \right) \mbox{}
		\]
	\end{example}
	
	\begin{definition}[$\delta_2$-finite functions]\label{def2} A function $f(z)$ is called $\delta_2$-finite if there exist polynomials $P_0(z),\cdots,P_d(z)$, not all zero, such that
		\begin{equation}\label{eq12}
			P_d(z)\delta_{2,z}^{d+2}\left(f(z)\right)+\cdots+P_2(z)\delta_{2,z}^4\left(f(z)\right)+P_1(z)\delta_{2,z}^3\left(f(z)\right)+	P_0(z)f(z)=0.
		\end{equation}
	\end{definition}
	
	\begin{remark} From the definition of a $\delta_2$-finite function one sees that it does not necessarily require the first $\delta_2$ derivative to be $1$ as given in $(\ref{eq8})$. Indeed, Definition \ref{def2} extends to the non-homogeneous case since by the Leibniz product rule, we can differentiate a non-homogeneous QDE finitely many times and get a homogeneous one. However, we mention that our Maxima implementation can compute non-homogeneous QDEs, although we avoid non-homogeneity for formal power series representations to escape dealing with constant terms.
	\end{remark}
	
	Although the differential equation $(\ref{eq12})$ can be linear (holonomic case), we often say that a $\delta_2$-finite function is a function that satisfies a homogeneous QDE with polynomial coefficients, or simply a homogeneous QDE. 
	
	\begin{theorem}\label{theo1} The class of holonomic functions is strictly contained in the class of $\delta_2$-finite functions.
	\end{theorem}
	\begin{proof} 
		The holonomic case is deduced from $(\ref{eq12})$ by observing that having $P_j(z)=0$, for $j\in$ $\{1,\ldots,d\}$ $\setminus$ $\left(i\cdot(i+3)/2\right)_{i\geq0}$, yields a holonomic differential equation since $\delta_{2,z}^{2+i\cdot(i+3)/2}(f(z))$ is a linear differential monomial for all non-negative integer $i$ (see $(\ref{eq8})$). \QEDA
	\end{proof}
	
	It is well-known that reciprocals of holonomic functions are generally not holonomic. Although the aim of this paper is not to study the closure properties of $\delta_2$-finite functions, in the following theorem, we show that an important family of reciprocals of holonomic functions is included in the class of $\delta_2$-finite functions.
	
	\begin{theorem}\label{theo3} The reciprocal of a function that satisfies a second-order holonomic differential equation is $\delta_2$-finite.
	\end{theorem}
	\begin{proof}
		Let $y(z)$ be a function that satisfies a second-order holonomic differential equation. Therefore there exist two rational functions $R_1(z)$ and $R_2(z)$ over $\mathbb{K}$ such that
		\begin{equation}\label{eq61}
			\frac{d^2}{dz^2}y(z) = R_1(z)\cdot \frac{d}{dz}y(z) + R_2(z)\cdot y(z).
		\end{equation}
		Let $u(z)=1/y(z)$. Without loss of generality, we look for a second-order QDE satisfied by $u$. This implies that the maximum order of the $\delta_2$ differentiation needed is $10$. We are looking for rational coefficients $c_j=c_j(z)$, $1\leq j \leq 9$ so that
		\begin{equation}\label{eq62}
			\sum_{j=2}^{10}c_{j-1}(z)\cdot \delta_{2,z}^j\left(u(z)\right)=0;
		\end{equation}
		the polynomial coefficients will be deduced by multiplication by the common denominator. Substituting $(\ref{eq61})$ into the left-hand side of $(\ref{eq62})$ and collecting the coefficients yields
		\begin{multline}\label{eq63}
			\frac{4 {c_9}\, {{\left( \frac{d}{d z} \operatorname{y}(z)\right) }^{4}}}{{{\operatorname{y}(z)}^{6}}}-\frac{2 \left( 2 {R_1}\, {c_9}+{c_8}\right) \, {{\left( \frac{d}{d z} \operatorname{y}(z)\right) }^{3}}}{{{\operatorname{y}(z)}^{5}}}+\frac{2 {c_6}\, {{\left( \frac{d}{d z} \operatorname{y}(z)\right) }^{2}}}{{{\operatorname{y}(z)}^{3}}}\\
			+\frac{\left( -4 {R_2}\, {c_9}+{{R}_{1}^{2}}\, {c_9}+{R_1}\, {c_8}+2 {c_7}+{c_5}\right) \, {{\left( \frac{d}{d z} \operatorname{y}(z)\right) }^{2}}}{{{\operatorname{y}(z)}^{4}}}+\frac{\left( -{R_1}\, {c_6}-{c_3}\right) \, \left( \frac{d}{d z} \operatorname{y}(z)\right) }{{{\operatorname{y}(z)}^{2}}}\\
			+\frac{\left( {R_2}\, \left( 2 {R_1}\, {c_9}+{c_8}\right) -{R_1}\, {c_7}-{c_4}\right) \, \left( \frac{d}{d z} \operatorname{y}(z)\right) }{{{\operatorname{y}(z)}^{3}}}+\frac{{c_1}-{R_2}\, {c_6}}{\operatorname{y}(z)}+\frac{{R_2}\, \left( {R_2}\, {c_9}-{c_7}\right) +{c_2}}{{{\operatorname{y}(z)}^{2}}}.
		\end{multline}
		Finally, we solve the linear system obtained by equating the rational coefficients in $(\ref{eq63})$ to zero. We obtain the solution
		\begin{equation}\label{eq64}
			\left\{{c_1}=0\operatorname{,}{c_2}={R_2}\, \mathit{C}\operatorname{,}{c_3}=0\operatorname{,}{c_4}=-{R_1}\, \mathit{C}\operatorname{,}{c_5}=-2 \mathit{C}\operatorname{,}{c_6}=0\operatorname{,}{c_7}=\mathit{C}\operatorname{,}{c_8}=0\operatorname{,}{c_9}=0\right\},
		\end{equation}
		where $\mathit{C}:=C(z)$ is an arbitrary rational function in $\mathbb{K}(z)$. Therefore, the reciprocal $u(z)$ of $y(z)$ satisfies a differential equation of the form
		\begin{equation}\label{eq65}
			\mathit{C}(z) \left( \operatorname{u}(z) \left( \frac{{{d}^{2}}}{d {{z}^{2}}} \operatorname{u}(z)\right) -2\, {{\left( \frac{d}{d z} \operatorname{u}(z)\right) }^{2}}-{R_1}(z)\, \operatorname{u}(z) \left( \frac{d}{d z} \operatorname{u}(z)\right) +{R_2}(z)\, {{\operatorname{u}(z)}^{2}} \right)=0
		\end{equation}
		Let $p(z)$ be the common denominator of $R_1(z)$ and $R_2(z)$. Then there exist $q(z)$ and $r(z)$ in $\mathbb{K}(z)$ such that $R_1(z)=-q(z)/p(z)$ and $R_2(z)=-r(z)/p(z)$. Equation $(\ref{eq65})$ is equivalent to
		\begin{equation}\label{eq66}
			{p}(z)\, \delta_{2,z}^8\left(\operatorname{u}(z)\right) -2\, {p}(z)\, \delta_{2,z}^6\left(\operatorname{u}(z)\right) +{q}(z)\,\delta_{2,z}^5\left(\operatorname{u}(z)\right) -{r}(z)\, \delta_{2,z}^3\left(\operatorname{u}(z)\right)=0,
		\end{equation}
		which concludes the proof. \QEDA
	\end{proof}
	
	\begin{corollary}\label{cor1} Let $f(z)$ be a holonomic function that satisfies the differential equation
		\begin{equation}\label{eq67}
			\operatorname{p}(z) \left( \frac{{{d}^{2}}}{d {{z}^{2}}} \operatorname{y}(z)\right) + \operatorname{q}(z) \left( \frac{d}{d z} \operatorname{y}(z)\right)  + \operatorname{r}(z) \operatorname{y}(z)=0,
		\end{equation}
		with $p(z)\neq0$. Then $1/f(z)$ is $\delta_2$-finite and satisfies the differential equation
		\begin{equation}\label{eq68}
			\operatorname{p}(z) \operatorname{y}(z) \left( \frac{{{d}^{2}}}{d {{z}^{2}}} \operatorname{y}(z)\right) -2 \operatorname{p}(z) {{\left( \frac{d}{d z} \operatorname{y}(z)\right) }^{2}}+\operatorname{q}(z) \operatorname{y}(z) \left( \frac{d}{d z} \operatorname{y}(z)\right) -\operatorname{r}(z) {{\operatorname{y}(z)}^{2}}=0.
		\end{equation}
	\end{corollary}
	\begin{proof}Immediate from the proof of Theorem \ref{theo3}. \QEDA
	\end{proof}
	
	\begin{example} Using Corollary \ref{cor1}, we can already give QDEs satisfied by $\sec(z)=1/\cos(z)$ and $\csc(z)=1/\sin(z)$. Since $\cos(z)$ and $\sin(z)$ satisfy the differential equation
		\begin{equation}\label{eq69}
			\frac{{{d}^{2}}}{d {{z}^{2}}} \operatorname{y}(z)+\operatorname{y}(z)=0,
		\end{equation}
		we have $p(z)=1$, $q(z)=0$ and $r(z)=1$, where $p,q,r$ are defined as in $(\ref{eq67})$. Therefore $\sec(z)$ and $\csc(z)$ satisfy the following differential equation
		\begin{equation}\label{eq70}
			\operatorname{y}(z) \left( \frac{{{d}^{2}}}{d {{z}^{2}}} \operatorname{y}(z)\right) -2 {{\left( \frac{d}{d z} \operatorname{y}(z)\right) }^{2}}-{{\operatorname{y}(z)}^{2}}=0.
		\end{equation}
	\end{example}
	
	In the next section, we present a general strategy to search for a QDE of least order satisfied by a given $\delta_2$-finite function.
	
	\section{Computing quadratic differential equations}\label{sec3}
	
	Given a $\delta_2$-finite function $f$, we propose an algorithm to compute the least-order QDE satisfied by $f$. This can be done by standard approaches after replacing the usual derivative operator by $\delta_{2,z}$. We consider the method of ansatz with undetermined coefficients as described in \cite{Koepf1992}.
	
	\begin{example}[$f(z):=1/\log(1+z)$]\item
		
		\begin{enumerate}
			\item $\delta_{2,z}^3(f(z))=1/\left(\log(1+z)\right)^2$, and we look for a rational function $C_0:=C_0(z)\in\mathbb{Q}(z)$ such that
			\begin{equation}\label{eq13}
				\delta_{2,z}^3(f(z)) + C_0(z) f(z) = \frac{1}{\left(\log(1+z)\right)^2} + C_0\frac{1}{\log(1+z)}= 0.
			\end{equation}
			Such a $C_0$ does not exist since $-\delta_{2,z}^3(f(z))/f(z)=-1/\log(1+z)\notin \mathbb{Q}(z)$\footnote{The aim is to collect rational factors that may appear while differentiating $f(z)$.}. In such a case the QDE sought might be of higher order.
			\item  $\delta_{2,z}^4(f(z))=-1/\left((1+z)\left(\log(1+z)\right)^2\right)$, and we look for rational functions $C_0,C_1\in\mathbb{Q}(z)$ such that
			\begin{equation}\label{eq14}
				\delta_{2,z}^4(f(z)) + C_1\delta_{2,z}^3(f(z)) + C_0f(z) = \frac{-1}{\left((1+z)\left(\log(1+z)\right)^2\right)} + \frac{C_1}{\left(\log(1+z)\right)^2} + \frac{C_0}{\log(1+z)} = 0,
			\end{equation}
			which is equivalent to
			\begin{equation}\label{eq15}
				\frac{C_0\left(1+z\right)\log(1+z)+C_1(1+z)-1}{\left(\log(1+z)\right)^2(1+z)}=0.
			\end{equation}
			We force the numerator to vanish by equating the coefficient in $\mathbb{Q}(z)[\log(1+z)]$ to zero. The obtained linear system is trivial and we get the solution
			\begin{equation}\label{eq16}
				\left\{\left(C_0=0,C_1=\frac{1}{1+z}\right)\right\}.
			\end{equation}
			Thus $f(z)$ satisfies the QDE
			\begin{equation}\label{eq17}
				\delta_{2,z}^4(y(z)) + \frac{1}{(1+z)}\delta_{2,z}^3(y(z))= \frac{d}{dz}y(z) + \frac{1}{(1+z)} y(z)^2=0.
			\end{equation}
			After clearing the denominators we finally get
			\begin{equation}\label{eq18}
				(1+z)\frac{d}{dz}y(z) + y(z)^2 = 0,
			\end{equation}
			with polynomial coefficients.
		\end{enumerate}
	\end{example}
	
	\begin{example}[$f(z):=\tan(z)$]\item
		
		\begin{enumerate}
			\item $\delta_{2,z}^3(f(z))=\left(\tan(z)\right)^2$, and we seek $C_0\in\mathbb{Q}(z)$ such that
			\begin{equation}\label{eq19}
				\delta_{2,z}^3(f(z)) + C_0(z) f(z) = \left(\tan(z)\right)^2 + C_0\tan(z)= 0.
			\end{equation}
			Since $-\delta_{2,z}^3(f(z))/f(z)=-\tan(z)\notin \mathbb{Q}(z)$, we proceed to the next step.
			\item $\delta_{2,z}^4(f(z))=1+\left(\tan(z)\right)^2$, and we seek $C_0,C_1\in\mathbb{Q}(z)$ such that
			\begin{equation}\label{eq20}
				\delta_{2,z}^4(f(z)) + C_1\delta_{2,z}^3(f(z)) + C_0f(z) = 1+\left(\tan(z)\right)^2 + C_1\left(\tan(z)\right)^2 + C_0\tan(z) = 0.
			\end{equation}
			The linear system obtained after equating the coefficients in $\mathbb{Q}(z)[\tan(z)]$ to zero has no solution. We proceed to the next $\delta_2$-derivative.
			\item This will continue until $\delta_{2,z}^7(f(z))=2\tan(z)(1+\left(\tan(z)\right)^2)$. We then look for the rational functions $C_0$,$\ldots$,$C_4$ such that
			\begin{multline}\label{eq21}
				\delta_{2,z}^7(f(z)) + C_4\delta_{2,z}^6(f(z)) +\cdots+ C_0f(z)= 2\tan(z)+2\left(\tan(z)\right)^3+C_4\left(1+2\left(\tan(z)\right)^2+\left(\tan(z)\right)^4\right)\\
				+C_3(\tan(z)+\left(\tan(z)\right)^3)+C_2\left(1+\left(\tan(z)\right)^2\right)+C_1\left(\tan(z)\right)^2 + C_0\tan(z)=0.
			\end{multline}
			This is equivalent to
			\begin{equation}\label{eq22}
				{C_4}\, {{\tan{(z)}}^{4}}+\left( {C_3}+2\right) \, {{\tan{(z)}}^{3}}+\left( 2 {C_4}+{C_2}+{C_1}\right) \, {{\tan{(z)}}^{2}}+\left( {C_3}+{C_0}+2\right)  \tan{(z)}+{C_4}+{C_2}=0.
			\end{equation}
			After equating the coefficients in $\mathbb{Q}(z)[\tan(z)]$ to zero we find the solution 
			\begin{equation}\label{eq23}
				\left\{\left(C_0=0,C_1=0,C_2=0,C_3=-2,C_4=0\right)\right\}.
			\end{equation}
			Therefore we get the QDE
			
			\vspace{-0.4cm}
			
			\begin{equation}\label{eq24}
				\delta_{2,z}^7(y(z))-2\delta_{2,z}^5(y(z))=\frac{{{d}^{2}}}{d {{z}^{2}}} \operatorname{y}(z)-2 \operatorname{y}(z) \left( \frac{d}{d z} \operatorname{y}(z)\right) =0.
			\end{equation}
		\end{enumerate}
	\end{example}
	
	The above examples show how our algorithm proceeds: the coefficients in $\mathbb{K}(z)[\alpha_1(z),\ldots,\alpha_k(z)]$, where the $\alpha_j(z)$, $1\leq j\leq k$ are transcendental functions verifying $\alpha_i(z)/\alpha_j(z) \notin \mathbb{K}(z)$, $1\leq i\neq j \leq k$, are collected by computing ratios of terms in the expansion of the ansatz; if the ratio is rational, then the corresponding terms are grouped together, otherwise they correspond to two distinct linear equations for the unknown coefficients of the differential equation sought. For instance, to get the factorization $\left( {C_3}+2\right) \, {{\tan{(z)}}^{3}}$ in $(\ref{eq22})$, the algorithm computes the ratio $(C_3 (\tan(z))^3)/ (2 (\tan(z))^3)=C_3/2\in\mathbb{Q}(z)$. Since there are no other rational ratios with $(2 (\tan(z))^3)$, the factor $\left( {C_3}+2\right)$ is an equation of the linear system in this case.
	
	Our algorithm iterates on the $\delta_2$ order of the QDE sought, and stops once it reaches a certain maximal order $N_{max}\in \mathbb{N}$. To be more precise, let us introduce the following notation.
	
	\begin{definition}[$\mathcal{M}_{\delta_2}(d)$] We denote by $\mathcal{M}_{\delta_2}(d)$, the maximum integer $n$ such that the $\delta_2$ derivative of $\delta_2$ order $n$ is a differential monomial of order $d$.
	\end{definition}
	
	\begin{example} \item
		\begin{itemize}
			\item $\mathcal{M}_{\delta_2}(1)=6$, since $\delta_{2,z}^7(f(z))=\frac{d^2}{dz^2}f(z)$ and $\delta_{2,z}^6(f(z))=\left(\frac{d}{dz}f(z)\right)^2$ is the last $\delta_2$ derivative of $f(z)$ of order $1$.
			\item $\mathcal{M}_{\delta_2}(2)=10$.
		\end{itemize}
	\end{example}
	
	\begin{proposition}\label{prop}For all positive integers $d$, $\mathcal{M}_{\delta_2}(d)=3+d(d+5)/2$.
	\end{proposition}
	
	We observed that the order $4$ is enough for finding QDEs satisfied by common non-holonomic functions. Therefore a suitable default value for the maximum number of iteration in our algorithm is $N_{max}=\mathcal{M}_{\delta_2}(4)=21$. Algorithm \ref{Algo1} gives the possibility to increase this value from the input. We mention that the method of ansatz with undetermined coefficients does not require to fix a maximum degree for the polynomial coefficients. The survey paper \cite{koepf1997} on the computation of holonomic differential equations present the advantage of this method over others for finding least-order differential equations.
	
	Our algorithm for finding homogeneous quadratic differential equations can be summarized as follows.
	
	\begin{algorithm}[h!]
		\caption{Searching for a QDE satisfied by a $\delta_2$-finite function $f$}
		\label{Algo1}
		\begin{algorithmic} 
			\Require A $\delta_2$-finite function $f(z)$, and optionally, a maximum order $d$ (default value $4$).
			\Ensure \texttt{FAIL} or a QDE of $\delta_2$ order at most $N_{max}=\mathcal{M}_{\delta_2}(d)$ over $\mathbb{K}$.
			\begin{enumerate}
				\item If $f=0$\footnotemark then the QDE is found and we stop.
				\item $f\neq0$, compute $C_0=\delta_{2,z}^3 f(z) / f(z),$
				\begin{itemize}
					\item[(1-a)] if $C_0\in\mathbb{K}(z)$ i.e $C_0(z)=P(z)/Q(z)$ where $P$ and $Q$ are polynomials, then we have found a QDE satisfied by $f$:
					$$Q(z)y(z)^2 - P(z)y(z) = 0.$$
					\item[(1-b)] If $C_0\notin \mathbb{K}(z)$, then go to 3.
				\end{itemize}
			\end{enumerate}
								\algstore{pause3b}
							\end{algorithmic}
						\end{algorithm}
						\clearpage 
						
						\begin{algorithm}
							\ContinuedFloat
							\caption{Searching for a QDE satisfied by a $\delta_2$-finite function $f$}
							\begin{algorithmic}
								\algrestore{pause3b}	
								\State
			\begin{enumerate}
				\setcounter{enumi}{2}
				\item $N_{max}:=\mathcal{M}_{\delta_2}(d)$; 
				\begin{itemize}
					\item[(3-a)] set $N:=2$;
					\item[(3-b)] compute $\delta_{2,z}^{N+2}f;$ 
					\item[(3-c)] expand the ansatz
					\begin{equation}\label{eq25}
						\delta_{2,z}^{N+2}(f(z)) + C_{N-1} \delta_{2,z}^{N+1}(f(z)) + \cdots + C_0 f(z) = \sum_{i=0}^{E}S_i,
					\end{equation}
					in elementary summands with $C_{N-1},\ldots, C_{0}$ as unknowns. $E\geqslant N$ is the total number of summands $S_i$ obtained after expansion.		
					\item[(3-d)]\label{Algo13d} For each pair of summands $S_i$ and $S_j$ $(0\leqslant i\neq j\leqslant E)$, group them additively together if $R(z)=S_i(z)/S_j(z)\in\mathbb{K}(z)$. These groups represent the linearly independent expressions whose rational coefficients are linear in the unknowns $C_0, C_1,\ldots, C_{N-1}$. Equating these coefficients to zero yields a linear system to solve in $\mathbb{K}(z)$. If the system has a non-empty set of solutions, then select a non-zero solution if it exists. If such a solution exists, then the \textit{step is successful}. The algorithm returns (and stop) the QDE obtained by substituting the values found for $C_0, C_1,\ldots, C_{N-1}$ in $(\ref{eq25})$ and multiplying the result by their common denominator. In the other cases, the \textit{step is not successful} and we move to the next step.
					\item[(3-e)] Increment $N$ ($N:=N+1$), and go back to (3-b), unless $N=N_{max}$.
				\end{itemize}
			\item Return \texttt{FAIL} ($N=N_{max}$, no QDE of order at most $d$ was found).
			\end{enumerate}	
		\end{algorithmic}
	\end{algorithm}
	
	
	\begin{remark}
		Algorithm \ref{Algo1} should not be seen as a decision algorithm for $\delta_2$-finiteness since such an algorithm is impossible to describe (see \cite{richardson1969some}). This can be justified by similar arguments from \cite{kauers2011} discussed in the case of holonomic functions. The capacity of the method also depends on how functions are encoded in the computer algebra system (CAS) used. For instance, our Maple implementation of Algorithm \ref{Algo1} finds a QDE of order $2$ for $\exp(2\arctanh(\sin(2z)/(1+\cos(2z))))$, whereas our Maxima implementation finds another one of order $3$. This is because Maple allows some automatic simplifications for compositions of $\exp$ and $\arctanh$ (see Subsection \ref{subsec3}). Thus, although Algorithm \ref{Algo1} is meant to find lowest-order QDEs, its capacity to detect $\delta_2$-finite functions is partly defined by the CAS used. Nevertheless, this happens only in rare cases because both Maple and Maxima implementations usually have identical results with lowest order possible.
	\end{remark}
	\footnotetext[5]{Only for trivial zero equivalences. A QDE of order zero $y=0$ would be the output.}
	
	Let us now present some results using our Maple implementation \texttt{FPS:-QDE} (or simply \texttt{QDE} if the \texttt{FPS} package is already loaded). Maple's \texttt{dsolve} command is used to solve the corresponding QDEs.
	
	\begin{example}\item
		
		\begin{lstlisting}
			> FPS:-QDE(sec(z),y(z))
		\end{lstlisting}
		
		\vspace{-0.2cm}
		
		\begin{dmath}\label{eq26}
			-y \! \left(z \right)^{2}-2 \left(\frac{d}{d z}y \! \left(z \right)\right)^{2}+\left(\frac{d^{2}}{d z^{2}}y \! \left(z \right)\right) y \! \left(z \right)=0
		\end{dmath}
		
		\vspace{-0.2cm}
		
		\begin{lstlisting}
			> dsolve(%,y(z))
		\end{lstlisting}
		
		\vspace{-0.2cm}
		
		\begin{dmath}\label{eq27}
			y \! \left(z \right)=\frac{1}{\textit{\_C1} \sin \! \left(z \right)-\textit{\_C2} \cos \! \left(z \right)}
		\end{dmath}
		
		\vspace{-0.3cm}
		
		\begin{lstlisting}
			> FPS:-QDE(z/log(1+z),y(z))
		\end{lstlisting}
		
		\vspace{-0.2cm}
		
		\begin{dmath}\label{eq28}
			\left(-1-z \right) y \! \left(z \right)+y \! \left(z \right)^{2}+z \left(1+z \right) \left(\frac{d}{d z}y \! \left(z \right)\right)=0
		\end{dmath}
		
		\vspace{-0.2cm}
		
		\begin{lstlisting}
			> dsolve(%,y(z))
		\end{lstlisting}
		
		\vspace{-0.2cm}
		
		\begin{dmath}\label{eq29}
			y \! \left(z \right)=\frac{z}{\ln \! \left(1+z \right)+\textit{\_C1}}
		\end{dmath}
		
		Next we compute a QDE for the generating function of {B}ernoulli polynomials of arbitrary order $k$ (see \cite{choi2014note}).
		\begin{lstlisting}
			> FPS:-QDE((t/(exp(t)-1))^k*exp(x*t),y(t))
		\end{lstlisting}
		
		\vspace{-0.2cm}
		
		\begin{dmath}\label{eq30}
			\left(t k x -t \,x^{2}+k^{2}-2 k x \right) y \! \left(t \right)^{2}+\left(-t k +2 x t +2 k \right) \left(\frac{d}{d t}y \! \left(t \right)\right) y \! \left(t \right)-t \left(1+k \right) \left(\frac{d}{d t}y \! \left(t \right)\right)^{2}+t k \left(\frac{d^{2}}{d t^{2}}y \! \left(t \right)\right) y \! \left(t \right)=0
		\end{dmath}
		
		\vspace{-0.2cm}
		
		\begin{lstlisting}
			> dsolve(%,y(t))
		\end{lstlisting}
		
		\vspace{-0.2cm}
		
		\begin{dmath}\label{eq31}
			y \! \left(t \right)=\frac{{\mathrm e}^{x t}}{\left(\frac{{\mathrm e}^{t} \textit{\_C1} -\textit{\_C2}}{t k}\right)^{k}}
		\end{dmath}
		
		\vspace{-0.2cm}
		
		\begin{lstlisting}
			> FPS:-QDE(sec(z)^k,y(z))
		\end{lstlisting}
		
		\vspace{-0.2cm}
		
		\begin{dmath}\label{eq32}
			-k^{2} y \! \left(z \right)^{2}+\left(-k -1\right) \left(\frac{d}{d z}y \! \left(z \right)\right)^{2}+k \left(\frac{d^{2}}{d z^{2}}y \! \left(z \right)\right) y \! \left(z \right)=0
		\end{dmath}
		
		\vspace{-0.2cm}
		
		\begin{lstlisting}
			> dsolve(%,y(z))
		\end{lstlisting}
		
		\vspace{-0.2cm}
		
		\begin{dmath}\label{eq33}
			y \! \left(z \right)=\frac{1}{\left(\frac{\textit{\_C1} \sin \left(z \right)-\textit{\_C2} \cos \left(z \right)}{k}\right)^{k}}
		\end{dmath}
		
		\vspace{-0.2cm}
		
		\begin{lstlisting}
			> FPS:-QDE(tan(z)^k,y(z))
		\end{lstlisting}
		
		\vspace{-0.2cm}
		
		\begin{dmath}\label{eq34}
			\left(20 k^{2}-24\right) \left(\frac{d}{d z}y \! \left(z \right)\right)^{2}+4 k^{2} \left(\frac{d^{2}}{d z^{2}}y \! \left(z \right)\right) y \! \left(z \right)+3 \left(k -2\right) \left(k +2\right) \left(\frac{d^{2}}{d z^{2}}y \! \left(z \right)\right)^{2}+\left(-4 k^{2}+6\right) \left(\frac{d^{3}}{d z^{3}}y \! \left(z \right)\right) \left(\frac{d}{d z}y \! \left(z \right)\right)+k^{2} \left(\frac{d^{4}}{d z^{4}}y \! \left(z \right)\right) y \! \left(z \right)=0
		\end{dmath}
		The \texttt{dsolve} command does not find an explicit solution for the latter QDE. However, we can verify that $\tan(z)^k$ is solution in the following way.
		
		\begin{lstlisting}
			> simplify(eval(subs(y(z)=tan(z)^k,lhs(%))))
		\end{lstlisting}
		
		\begin{dmath}\label{eq35}
			0
		\end{dmath}
	\end{example}
	
	Our implementation is available in Maple 2022 as a feature of the \texttt{FindODE} command.
	
	\section{Power series representations of $\delta_2$-finite functions}\label{sec4}
	
	\subsection{QDE to QRE}
	
	In this section, we assume that the power series of $f(z)$ is represented as $\sum_{n=0}^{\infty} a_nz^n$. Note that the Laurent series case easily follows since the appropriate shift of initial values can be deduced from the coefficients of the computed differential equation (see \cite{BTphd}, \cite{teguia2021symbolic}). For any constant $x$ and a non-negative integer $k$, $(z)_0=1$ and $(z)_k=x\cdot(x+1)\cdots(x+k-1)$ denotes the Pochhammer symbol or shifted factorial.
	
	We need a rewrite rule similar to that of holonomic equations to convert differential equations into recurrence equations. It is important to remind what the algorithm looks like in the linear case. Therefore we recall how it works below (see \cite{Koepf1992} and \cite{salvy1994gfun}).
	
	\begin{equation}\label{eq36}
		z^p\cdot f^{(j)}    \longrightarrow  (n+1-p)_j\cdot a_{n+j-p}.
	\end{equation}
	
	In the present case, we also need to convert every differential monomial
	\begin{equation}\label{eq37}
		z^p \cdot f(z)^{(i)}\cdot f(z)^{(j)}, ~\text{for non-negative integers }~ i,j,p.
	\end{equation}
	into a term of the QRE sought.
	
	For all non-negative integers $i$, we have
	\begin{equation}\label{eq38}
		f(z)^{(i)} = \sum_{n=0}^{\infty} (n+1)_i\cdot a_{n+i} \cdot z^n,
	\end{equation}
	therefore
	\begin{eqnarray}
		f(z)^{(i)}\cdot f(z)^{(j)} &=& \left(\sum_{n=0}^{\infty} (n+1)_i\cdot a_{n+i} \cdot z^n \right) \cdot \left(\sum_{n=0}^{\infty} (n+1)_j\cdot a_{n+j} \cdot z^n\right)\nonumber\\
		&=& \sum_{n=0}^{\infty} \left( \sum_{k=0}^{n} (k+1)_i\cdot a_{k+i} \cdot (n-k+1)_j\cdot a_{n-k+j}\right) \cdot z^n. \label{eq39}
	\end{eqnarray}
	by application of the Cauchy product formula which introduces the dummy variable $k$. Finally multiplying $(\ref{eq39})$ by $z^p$ yields
	\begin{equation}\label{eq40}
		z^p \cdot f(z)^{(i)}\cdot f(z)^{(j)} = \sum_{n=0}^{\infty} \left(\sum_{k=0}^{n-p}(k+1)_i\cdot (n-p-k+1)_j\cdot a_{k+i}\cdot a_{n-p-k+j}\right) \cdot z^n,
	\end{equation}
	from which we deduce the rewrite rule
	
	\begin{equation}\label{eq41}
		z^p \cdot f(z)^{(i)}\cdot f(z)^{(j)} \longrightarrow \left(\sum_{k=0}^{n-p}(k+1)_i\cdot (n-p-k+1)_j\cdot a_{k+i}\cdot a_{n-p-k+j}\right).
	\end{equation}
	
	Thus a procedure to convert QDE into QRE follows immediately, i.e. use $(\ref{eq36})$ for linear differential monomials, and $(\ref{eq41})$ for quadratic ones. Our packages contain the procedure \texttt{FindQRE} to compute a QRE satisfied by the power series coefficients of a given $\delta_2$-finite function. We give a few examples computed by our Maxima implementation.
	
	\vspace{-0.1cm}
	
	\begin{example}\label{eg6}\item \vspace{-0.2cm}
		
		\noindent
		\begin{minipage}[t]{8ex}\color{red}\bf
			\begin{verbatim}
				(%i1) 
			\end{verbatim}
		\end{minipage}
		\begin{minipage}[t]{\textwidth}\color{blue}
			\begin{verbatim}
				FindQRE(tan(z),z,a[n]);
			\end{verbatim}
			\vspace{-0.3cm}
		\end{minipage}
		\definecolor{labelcolor}{RGB}{100,0,0}
		\[\displaystyle
		\parbox{10ex}{$\color{labelcolor}\mathrm{\tt (\%o1) }\quad $}
		\left( 1+n\right) \cdot \left( 2+n\right) \cdot {{a}_{n+2}}-2\cdot \sum_{k=0}^{n}\left( k+1\right) \cdot {{a}_{k+1}}\cdot {{a}_{n-k}}=0
		\]
		
		\noindent
		\begin{minipage}[t]{8ex}\color{red}\bf
			\begin{verbatim}
				(%i2) 
			\end{verbatim}
		\end{minipage}
		\begin{minipage}[t]{\textwidth}\color{blue}
			\begin{verbatim}
				FindQRE(z/(exp(z)-1),z,a[n]);
			\end{verbatim}
		\end{minipage}
		\definecolor{labelcolor}{RGB}{100,0,0}
		\[\displaystyle
		\parbox{10ex}{$\color{labelcolor}\mathrm{\tt (\%o2) }\quad $}
		\left( \sum_{k=0}^{n}{{a}_{k}}\cdot {{a}_{n-k}}\right) +\left( n-1\right) \cdot {{a}_{n}}+{{a}_{n-1}}=0\mbox{}
		\]
		
		\noindent
		\begin{minipage}[t]{8ex}\color{red}\bf
			\begin{verbatim}
				(%i3) 
			\end{verbatim}
		\end{minipage}
		\begin{minipage}[t]{\textwidth}\color{blue}
			\begin{verbatim}
				FindQRE(log(1+sin(z)),z,a[n]);
			\end{verbatim}
		\end{minipage}
		\definecolor{labelcolor}{RGB}{100,0,0}
		\vspace{-0.5cm}
		\begin{multline*}
			\displaystyle
			\parbox{10ex}{$\color{labelcolor}\mathrm{\tt (\%o3) }\quad $}\hspace{-0.5cm}
			\left( \sum_{k=0}^{n}\left( k+1\right) \cdot \left( k+2\right) \cdot {{a}_{k+2}}\cdot \left( n-k+1\right) \cdot {{a}_{n-k+1}}\right)+\left( 1+n\right) \cdot \left( 2+n\right) \cdot \left( 3+n\right) \cdot {{a}_{n+3}}=0
		\end{multline*}
	\end{example}
	
	\subsection{Normal forms for $\delta_2$-finite power series}
	
	The last step of our computations consists of using the obtained QRE to write the highest order indexed variable in terms of the others. Evaluating the recurrence equation at some integers allows to reveal the initial values to be computed from evaluation of the input function and its derivatives. Usually, the summations coming from Cauchy products have to be evaluated at their lower and upper bounds in order to extract all occurrences of the highest order indexed variable. Let us consider the QRE obtained for $z/(\exp(z)-1)$.
	\begin{equation}\label{eq42}
		\left( \sum_{k=0}^{n}{{a}_{k}}\cdot {{a}_{n-k}}\right) +\left( n-1\right) \cdot {{a}_{n}}+{{a}_{n-1}}=0\mbox{}.
	\end{equation}
	There is only one summation and we need to extract $a_n$ from it. This corresponds to the indices $k=0$, $k=n$, and the extra second summand. We get
	\begin{equation}\label{eq43}
		\left( \sum_{k=1}^{n-1}{{a}_{k}}\cdot {{a}_{n-k}}\right) +   2\cdot a_0\cdot a_n + \left( n-1\right) \cdot {{a}_{n}}+{{a}_{n-1}}=0\mbox{}.
	\end{equation}
	We then substitute the value of $a_0$ and deduce the recursive formula sought. Together with the necessary initial values, the obtained representation defines a unique sequence (see $(\ref{eq3})$) of coefficients which characterizes the power series expansion of $z/(\exp(z)-1)$. The obtained recursive formula is
	
	\begin{equation}\label{eq44}
		a_{n+3} = - \frac{a_{n+2} + \sum_{k=1}^{n+2}a_k\cdot a_{n+3-k}}{n+4},~ n\geqslant 0,~ a_0=1, a_1=-1/2, a_2=-1/12,
	\end{equation}
	which can be used to recover a well-known Ramanujan identity for the {B}ernoulli numbers (see \cite{chellali1988acceleration}). Indeed, substituting $a_n$ by $B_n/n!$, where $B_n$ denotes the $n^{\text{th}}$ {B}ernoulli number, and multiplying both sides of $(\ref{eq44})$ by $(n+3)!$, gives
	\begin{equation}\label{eq45}
		B_{n+3} = - \frac{(n+3)B_{n+2}+\sum_{k=1}^{n+2}\binom{n+3}{k} B_k\cdot B_{n+3-k}}{n+4}.
	\end{equation}
	Then using the known fact (which can also be deduced) that except $B_1$, all {B}ernoulli numbers of odd subscripts are zero, it follows that $B_{n+3}$ and $B_{n+2}$ cannot be both non-zero at the same time. Finally substituting $n$ by $2n-1$ leads to the identity
	
	\begin{eqnarray}\label{eq46}
		B_{2n+2} &=& -\frac{1}{2n+3}\sum_{k=1}^{n}\binom{2n+2}{2k} B_{2k}\cdot B_{2(n+1-k)}\nonumber\\
		&=& -\frac{1}{2n+3}\left(2\cdot\sum_{k=1}^{\lceil n/2 \rceil}\binom{2n+2}{2k} B_{2k}\cdot B_{2(n+1-k)}-\binom{2n+2}{n+1}B_{n+1}^2\right)~ (n\geqslant 1)
	\end{eqnarray}
	
	This shows how our algorithm can be of good help for manipulating {B}ernoulli numbers and similar sequences. 
	
	It is worth to ask what is the number of initial values required to have a valid formula for the coefficients of a formal power series. The case of the exponential generating function of Bernoulli numbers may seem obvious in this regard. To show that this is always possible, let us consider the representations for $(\cos(z))^2$ and $(\sin(z))^2$ which satisfy the same QDE and therefore the same QRE, but lead to two different recursive formulas for their power series. The differential equation found by Algorithm \ref{Algo1} is given by
	\begin{equation}\label{eq71}
		{{\left( \frac{d}{d z} \operatorname{y}(z)\right) }^{2}}+4 {{\operatorname{y}(z)}^{2}}-4 \operatorname{y}(z)=0,
	\end{equation}
	which lead to the following recurrence equation after application of the rewrite rules $(\ref{eq36})$ and $(\ref{eq41})$.
	\begin{equation}\label{eq72}
		\sum_{k=0}^{n}{\left. \left( k+1\right) \, \left( n-k+1\right) \, {a_{k+1}}\, {a_{n-k+1}}\right.} +4\,  \sum_{k=0}^{n}{\left. {a_k}\, {a_{n-k}}\right.} -4 {a_n}=0.
	\end{equation}
	Note that these two functions satisfy a third-order holonomic differential equation. Algorithm \ref{Algo1} finds $(\ref{eq71})$ because it is of lower order.
	
	Now we extract the highest order indexed variable $a_{n+1}$ and $a_n$. Equation $(\ref{eq72})$ is equivalent to
	\begin{equation}\label{eq73}
		{\left. 2\, {a_1} \, (n+1) \, {a_{n+1}} \right.} + {\left. 4\, {a_2} \, {n} \, {a_n} \right.} + {\left. 4\,\left( 2\, {a_0} - 1\right) \, {a_n}  \right.} + \sum_{k=2}^{n-2}{\left. \left( k+1\right) \, \left( n-k+1\right) \, {a_{k+1}}\, {a_{n-k+1}}\right.} +4\,  \sum_{k=1}^{n-1}{\left. {a_k}\, {a_{n-k}}\right.}=0.
	\end{equation}
	Using the initial coefficients $a_0=(\cos(0))^2=1$, $a_1=\left((\cos(z))^2\right)'(0)=0$ and $a_2=\left((\cos(z))^2\right)''(0)/2!=-1$, one easily deduce a formula for $a_n$ which uniquely identifies the coefficients of the power series of $(\cos(z))^2$. For $(\sin(z))^2$, we have $a_0=a_1=0$ and $a_2=1$, which also yields a formula for $a_n$ from $(\ref{eq73})$ and therefore uniquely defines the power series coefficients of $(\sin(z))^2$.
	
	As one can notice, the identification of a $\delta_2$-finite function to its representation relies on the uniqueness of its power series, which is guaranteed by the analytic property of that function in the neighborhood considered.
	
	Regarding the required initial values, observe that these are always deduced from extractions of highest order indexed variables from the summations appearing in the computed QRE. Therefore we should define a bound for the number of extractions required to make sure that these extractions do not continue indefinitely when several initial values are zero. Such a bound can be easily determined from the coefficients of terms with summations in the QRE. Indeed, the rewrite rule $(\ref{eq41})$ leaves the constant factors out of the summation. Moreover, in this particular case we may write the power series as $f(z)=z^p\sum_{n=0}^{\infty} a_n\cdot z^n$, $p>0$, $a_0\neq0$. By differentiation, we see that $p$ appear as a multiplicative factor starting from the first derivative. Algorithmically, the integer part of the maximum absolute value of constant factors appearing in front of summations in the QRE constitute a bound for the number of required initial values. However, this upper bound may seem crude for some examples; one could also determine a bound for $p$ by using known methods to calculate Laurent series solutions of algebraic differential equations (see \cite{N2017}).
	
	We are now all set to give the following theorem, which summarizes our results.
	
	\begin{theorem}\label{theo2} Given a $\delta_2$-finite function $f(z)$, the following steps
		\begin{enumerate}
			\item Use Algorithm \ref{Algo1} to compute a QDE satisfied by $f(z)$;
			\item Expand the left-hand side of the obtained differential equation and convert it into a QRE using the rewrite rules $(\ref{eq36})$ and $(\ref{eq41})$;
			\item Use the obtained QRE to write its highest-order indexed variable in terms of the preceding ones with the required initial values;
		\end{enumerate}
		define a normal form\footnote{The used algorithm is then called a normal function (see \cite[Chapter 3]{geddes1992algorithms}.)} of the power series representation of $f(z)$.
	\end{theorem}
	\begin{proof}
		First, we need to show that these three steps define a normal transformation (or normal function) for the class of $\delta_2$-finite functions. For that purpose we must prove that the output representation and the input function define the same mathematical object, and that $\delta_2$-finite functions equivalent to zero have the same representation (see \cite[Def 3.1 \& Def 3.3]{geddes1992algorithms}). These follow from the arguments developed in the previous paragraphs of this section. By fixing the bound for the number of necessary initial values, the mathematical object is uniquely determined and zero-equivalences are automatically detected.
		
		The second part of the proof consists of showing that the representation given by the normal transformation used is unchanged under the application of that transformation to it. That is, of course, the case since the differential equation and the recurrence equation associated with the recursive formulas obtained from the three steps in Theorem \ref{theo2} characterize it. Thus, one can see the transformation as the one that applies these three steps if the given function is not in the desired output form, and returns the input if it already has the desired form. \QEDA
	\end{proof}
	
	Normal forms are often used in computer algebra to represent mathematical objects. Our method is an extension of common techniques used to prove identities between holonomic functions (see \cite{zeilberger1990holonomic}). The most interesting point of Theorem \ref{theo2} is that zero-equivalences between $\delta_2$-finite functions can be detected with our algorithm. This is an important fact for it is well-noticeable that identities between non-holonomic functions are not easily detected from an algorithmic perspective. Before proving some identities, let us give examples of representations obtained with our Maple implementation.

	\begin{example}The argument \texttt{fpstype=quadratic} is specified to apply the method to non-holonomic functions directly.
		
		\begin{small}
			\begin{lstlisting}
				> FPS(tan(z),z,n,fpstype=quadratic)
			\end{lstlisting}
			
			\vspace{-0.3cm}
			
			\begin{dmath}\label{eq47}
				\hspace{-1.4cm}\mathit{Series} \! \left(\left[\moverset{\infty}{\munderset{n =0}{\textcolor{gray}{\sum}}}\! A \! \left(n \right) z^{n}, A \! \left(n +3\right) \hiderel{=} -\frac{-2 A \! \left(n +1\right)+\moverset{n}{\munderset{\textit{\_k} =1}{\textcolor{gray}{\sum}}}\! \left(-2 \left(\textit{\_k} +1\right) A \! \left(\textit{\_k} +1\right) A \! \left(n -\textit{\_k} +1\right)\right)}{\left(n +2\right) \left(n +3\right)}\right],
				\\[-0.35cm]
				\left\{A \! \left(n \right)\right\},\left\{A \! \left(0\right) \hiderel{=} 0,A \! \left(1\right) \hiderel{=} 1,A \! \left(2\right) \hiderel{=} 0\right\},\mathit{INFO} \right)
			\end{dmath}
			\vspace{-0.2cm}
			
			\begin{lstlisting}
				> FPS(1/(1+sin(z)),z,n,fpstype=quadratic)
			\end{lstlisting}
			
			\vspace{-0.3cm}
			
			\begin{dmath}\label{eq48}
				\hspace{-0.8cm}\mathit{Series} \! \left(\left[\moverset{\infty}{\munderset{n =0}{\textcolor{gray}{\sum}}}\! A \! \left(n \right) z^{n},	A \! \left(n +2\right)\hiderel{=}-\frac{-5 A \! \left(n \right)+\moverset{n -1}{\munderset{\textit{\_k} =1}{\textcolor{gray}{\sum}}}\! \left(-3 A \! \left(\textit{\_k} \right) A \! \left(n -\textit{\_k} \right)\right)}{\left(n +1\right) \left(n +2\right)}\right],\\				
				\left\{A \! \left(n \right)\right\},\left\{A \! \left(0\right) \hiderel{=} 1,A \! \left(1\right) \hiderel{=} -1\right\},\mathit{INFO} \right)
			\end{dmath}
		\end{small}
	\end{example} 
	
	In particular, these obtained representations can be used to deduce truncated series. Our Maxima package contains the procedure \texttt{QTaylor}, implemented for that purpose. Below we compare the results with the built-in Maxima \texttt{taylor} command.
	
	\begin{example}\item
		
		\noindent
		\begin{minipage}[t]{8ex}\color{red}\bf
			\begin{verbatim}
				(%i1) 
			\end{verbatim}
		\end{minipage}
		\begin{minipage}[t]{\textwidth}\color{blue}
			\begin{verbatim}
				taylor(sec(z),z,0,7);
			\end{verbatim}
		\end{minipage}
		\definecolor{labelcolor}{RGB}{100,0,0}
		\[\displaystyle
		\parbox{10ex}{$\color{labelcolor}\mathrm{\tt (\%o1) }\quad $}
		1+\frac{{{z}^{2}}}{2}+\frac{5 {{z}^{4}}}{24}+\frac{61 {{z}^{6}}}{720}+\operatorname{...}
		\]
		
		\noindent
		\begin{minipage}[t]{8ex}\color{red}\bf
			\begin{verbatim}
				(%i2) 
			\end{verbatim}
		\end{minipage}
		\begin{minipage}[t]{\textwidth}\color{blue}
			\begin{verbatim}
				QTaylor(sec(z),z,0,7);
			\end{verbatim}
		\end{minipage}
		\definecolor{labelcolor}{RGB}{100,0,0}
		\[\displaystyle
		\parbox{10ex}{$\color{labelcolor}\mathrm{\tt (\%o2) }\quad $}
		\frac{61 {{z}^{6}}}{720}+\frac{5 {{z}^{4}}}{24}+\frac{{{z}^{2}}}{2}+1
		\]
		
		\noindent
		\begin{minipage}[t]{8ex}\color{red}\bf
			\begin{verbatim}
				(%i3) 
			\end{verbatim}
		\end{minipage}
		\begin{minipage}[t]{\textwidth}\color{blue}
			\begin{verbatim}
				taylor(tan(z),z,0,7);
			\end{verbatim}
		\end{minipage}
		\definecolor{labelcolor}{RGB}{100,0,0}
		\[\displaystyle
		\parbox{10ex}{$\color{labelcolor}\mathrm{\tt (\%o3) }\quad $}
		z+\frac{{{z}^{3}}}{3}+\frac{2 {{z}^{5}}}{15}+\frac{17 {{z}^{7}}}{315}+\operatorname{...}
		\]
		
		\noindent
		\begin{minipage}[t]{8ex}\color{red}\bf
			\begin{verbatim}
				(%i4) 
			\end{verbatim}
		\end{minipage}
		\begin{minipage}[t]{\textwidth}\color{blue}
			\begin{verbatim}
				QTaylor(tan(z),z,0,7);
			\end{verbatim}
		\end{minipage}
		\definecolor{labelcolor}{RGB}{100,0,0}
		\[\displaystyle
		\parbox{10ex}{$\color{labelcolor}\mathrm{\tt (\%o4) }\quad $}
		\frac{17 {{z}^{7}}}{315}+\frac{2 {{z}^{5}}}{15}+\frac{{{z}^{3}}}{3}+z
		\]
	\end{example}
	The ability to do such calculations further sustains our algorithm.
	
	\subsection{Proving identities}\label{subsec3}
	
	As a consequence of Theorem \ref{theo2}, we present automatic proofs of two non-holonomic identities. Of course, we do not ignore Richardson's theorem (see \cite{richardson1969some}). However, it is clear that bringing the zero-equivalence problem to the class of power series solves this issue for the class under consideration. Therefore for $\delta_2$-finite functions our approach is a decision procedure for zero-equivalence. Two expressions $A$ and $B$ define the same $\delta_2$-finite function (at least in the neighborhood of the point of expansion), if our algorithm finds the same power series representation for both of them, or if the one of $A-B$ is zero. The latter comes as a conclusion when all the necessary initial values of the representation sought are zero. It should be noted that our algorithm does not check in which disk the identity is valid. However for analytic functions, if they have the same power series, then they are identical in the largest possible disk. Regarding power series representations, we focus on the neighborhood of the origin.
	
	As first identity, consider
	\begin{equation}
		\log\left(\tan\left(\dfrac{z}{2}\right) + \sec\left(\dfrac{z}{2}\right)\right) = \arcsinh\left(\dfrac{\sin(z)}{1+\cos(z)}\right),~~ -\pi < z < \pi \label{eq49} 
	\end{equation}
	from \cite[Section 3.3]{geddes1992algorithms} (see also \cite[Exercise 9.8]{koepf2021computer}). Let $f$ be the left-hand side of $(\ref{eq49})$ and $g$ its right-hand side. Algorithm \ref{Algo1} finds the same differential equation for $f$, $g$, and $f-g$, which is
	
	\begin{lstlisting}
		> f:=log(tan(z/2)+sec(z/2)):
		> g:=arcsinh(sin(z)/(cos(z)+1)): 
		> FPS:-QDE(f,y(z))
	\end{lstlisting}
	
	\vspace{-0.3cm}
	
	\begin{dmath}\label{eq50}
		-\left(\frac{d}{d z}y \! \left(z \right)\right)^{2}-8 \left(\frac{d^{2}}{d z^{2}}y \! \left(z \right)\right)^{2}+4 \left(\frac{d^{3}}{d z^{3}}y \! \left(z \right)\right) \left(\frac{d}{d z}y \! \left(z \right)\right)=0.
	\end{dmath}
	
	We get the following power series representation for $f$
	\begin{small}
		\begin{lstlisting}
			> FPS(f,z,n,fpstype=quadratic)
		\end{lstlisting}
		
		\vspace{-0.2cm}
		
		\begin{dmath}\label{eq51}
			\mathit{Series} \! \left(\left[\moverset{\infty}{\munderset{n =0}{\textcolor{gray}{\sum}}}\! A \! \left(n \right) z^{n} , A \! \left(n +4\right) \hiderel{=} -\frac{1}{2 \left(n +2\right) \left(n +3\right) \left(n +4\right)}\! \left(-\frac{\left(n +2\right) A \left(n +2\right)}{2}+\\
			\left(\moverset{n}{\munderset{\textit{\_k} =1}{\textcolor{gray}{\sum}}}\! 4 \left(\textit{\_k} +1\right) \left(\textit{\_k} +2\right) \left(\textit{\_k} +3\right) A \! \left(\textit{\_k} +3\right) \left(n -\textit{\_k} +2\right) A \! \left(n -\textit{\_k} +2\right)\right)\\
			+\moverset{n}{\munderset{\textit{\_k} =1}{\textcolor{gray}{\sum}}}\! \left(-\left(\textit{\_k} +1\right) A \! \left(\textit{\_k} +1\right) \left(n -\textit{\_k} +2\right) A \! \left(n -\textit{\_k} +2\right)\right)+\\
			\moverset{n}{\munderset{\textit{\_k} =1}{\textcolor{gray}{\sum}}}\! \left(-8 \left(\textit{\_k} +1\right) \left(\textit{\_k} +2\right) A \! \left(\textit{\_k} +2\right) \left(n -\textit{\_k} +2\right) \left(n +3-\textit{\_k} \right) A \! \left(n +3-\textit{\_k} \right)\right)\right)\right],\\
			\left\{A \! \left(n \right)\right\},\left\{A \! \left(0\right) \hiderel{=} 0,A \! \left(1\right) \hiderel{=} \frac{1}{2},A \! \left(2\right) \hiderel{=} 0,A \! \left(3\right) \hiderel{=} \frac{1}{48}\right\},\mathit{INFO} \right).
		\end{dmath}
	\end{small}
	Thus the series representation of $f$ should differ from that of $g$ only by the initial values (like for $\sin(z)$ and $\cos(z)$). However, since the four necessary initial values are identical as shown by the following truncated series expansion, we deduce that $f=g$. The series representation of $f$ is obtained by substituting $f$ by $g$ in $(\ref{eq51})$. 
	
	\begin{lstlisting}
		> series(f-g,z,3)
	\end{lstlisting}
	
	\vspace{-0.3cm}
	
	\begin{dmath}\label{eq52}
		\mathrm{O}\! \left(z^{4}\right)
	\end{dmath}
	Furthermore, our algorithm detects this identity directly by finding zero as the power series representation of $f-g$.
	\begin{lstlisting}
		> FPS(f-g,z,n,fpstype=quadratic)
	\end{lstlisting}
	
	\vspace{-0.3cm}
	
	\begin{dmath}\label{eq53}
		0
	\end{dmath}
	
	One should note that the current Maple \texttt{simplify} command seems to be unable to recognize this zero-equivalence. 
	
	Our second identity is a similar one given by (see \cite[Section 9.1]{koepf2021computer})
	
	\begin{equation}
		\log\left(\dfrac{1+\tan(z)}{1-\tan(z)}\right) = 2 \arctanh\left(\dfrac{\sin(2z)}{1+\cos(2z)}\right),~~-\frac{\pi}{4}< z <\frac{\pi}{4}. \label{eq54}
	\end{equation}
	
	As previously we denote by $f$ and $g$ the left-hand side and the right-hand side of $(\ref{eq54})$, respectively. This identity can be recognized in Maple as follows.
	
	\begin{lstlisting}
		> f:=log((1+tan(z))/(1-tan(z))): 
		> g:=2*arctanh(sin(2*z)/(1+cos(2*z))):
		> simplify(exp(f)-exp(g))
	\end{lstlisting}
	
	\vspace{-0.3cm}
	
	\begin{dmath}\label{eq55}
		0
	\end{dmath}
	
	Indeed, the composition between $\exp$ and $\arctanh$ applies ``non-trivial'' simplifications that ease the work for the \texttt{simplify} command. 
	
	Let us nevertheless prove $(\ref{eq54})$ using our algorithm. Both sides satisfy the same differential equation, and have the same initial values for their power series representations.
	
	\begin{lstlisting}
		> series(f-g,z,3)
	\end{lstlisting}
	
	\vspace{-0.3cm}
	
	\begin{dmath}\label{eq56}
		\mathrm{O}\! \left(z^{4}\right)
	\end{dmath}
	
	\vspace{-0.4cm}
	
	\begin{lstlisting}
		> FPS:-QDE(g,y(z))
	\end{lstlisting}
	
	\vspace{-0.3cm}
	
	\begin{dmath}\label{eq57}
		-4 \left(\frac{d}{d z}y \! \left(z \right)\right)^{2}-2 \left(\frac{d^{2}}{d z^{2}}y \! \left(z \right)\right)^{2}+\left(\frac{d^{3}}{d z^{3}}y \! \left(z \right)\right) \left(\frac{d}{d z}y \! \left(z \right)\right)=0
	\end{dmath}
	
	This is enough to deduce that $f$ and $g$ coincide. We compute the power series representation from $g$ below.
	
	\begin{small}
		\begin{lstlisting}
			> FPS(g,z,n,fpstype=quadratic)
		\end{lstlisting}
		
		\vspace{-0.4cm}
		
		\begin{dmath}\label{eq58}
			\mathit{Series} \! \left(\left[\moverset{\infty}{\munderset{n =0}{\textcolor{gray}{\sum}}}\! A \! \left(n \right) z^{n}, A \! \left(n +4\right)\hiderel{=}-\frac{1}{2 \left(n +2\right) \left(n +3\right) \left(n +4\right)}\! \left(-8 \left(n +2\right) A \! \left(n +2\right)+\\
			\left(\moverset{n}{\munderset{\textit{\_k} =1}{\textcolor{gray}{\sum}}}\! \left(\textit{\_k} +1\right) \left(\textit{\_k} +2\right) \left(\textit{\_k} +3\right) A \! \left(\textit{\_k} +3\right) \left(n -\textit{\_k} +2\right) A \! \left(n -\textit{\_k} +2\right)\right)+\\
			\moverset{n}{\munderset{\textit{\_k} =1}{\textcolor{gray}{\sum}}}\! \left(-4 \left(\textit{\_k} +1\right) A \! \left(\textit{\_k} +1\right) \left(n -\textit{\_k} +2\right) A \! \left(n -\textit{\_k} +2\right)\right)\\
			+\moverset{n}{\munderset{\textit{\_k} =1}{\textcolor{gray}{\sum}}}\! \left(-2 \left(\textit{\_k} +1\right) \left(\textit{\_k} +2\right) A \! \left(\textit{\_k} +2\right) \left(n -\textit{\_k} +2\right) \left(n +3-\textit{\_k} \right) A \! \left(n +3-\textit{\_k} \right)\right)\right)\right],\\
			\left\{A \! \left(n \right)\right\},\left\{A \! \left(0\right) \hiderel{=} 0,A \! \left(1\right) \hiderel{=} 2,A \! \left(2\right) \hiderel{=} 0,A \! \left(3\right) \hiderel{=} \frac{4}{3}\right\},\mathit{INFO} \right).
		\end{dmath}
	\end{small}
	
	For this zero-equivalence, our Maple implementation detects that $f-g$ is a constant earlier in the computations.
	\begin{lstlisting}
		> FPS:-QDE(f-g,y(z))
	\end{lstlisting}
	
	\vspace{-0.3cm}
	
	\begin{dmath}\label{eq59}
		\frac{d}{d z}y \! \left(z \right)=0
	\end{dmath}
	Thus, unlike the previous identity which takes about a minute to be proven from $f-g$, here the proof is almost instantaneous.
	\begin{lstlisting}
		> FPS(f-g,z,n,fpstype=quadratic)
	\end{lstlisting}
	
	\vspace{-0.2cm}
	
	\begin{dmath}\label{eq60}
		0
	\end{dmath}
	
	\section{Conclusion}
	
	 We have proposed a general-purpose method to compute normal forms for representing the power series of a large class of non-holonomic functions, including that of holonomic functions. The computations were mainly presented for Laurent series and we believe that the method easily adapts to Puiseux series. One could incorporate ideas from \cite{Frankenphd} in this regard. Often the representations found may look ``big''. However, these are likely what one will find when calculating the series directly by hand if one wants to get complete formulas. Moreover, an important advantage of the proposed method is that it can simplify non-trivial identities. Previous Maple releases did not include formal power series computation for non-holonomic functions. We are delighted and grateful that our implementation available at
	
	\begin{center}
		\url{http://www.mathematik.uni-kassel.de/~bteguia/FPS_webpage/FPS.htm}
	\end{center}	
	
	is now incorporated into Maple 2022 in the \texttt{FormalPowerSeries} (renewed by the FPS package) and \texttt{DEtools} packages.

\end{document}